\documentclass[submission]{llncs}
\usepackage{amsmath,amssymb}
\usepackage{tikz}
\usetikzlibrary{positioning,calc}
\usepackage{algorithmicx}
\usepackage[noend]{algpseudocode}
\usepackage[small,it]{caption}

\let\doendproof\endproof
\renewcommand\endproof{~\hfill$\qed$\doendproof}

\title{Biclique coverings, rectifier networks and the cost of $\varepsilon$-removal}
\author{
Szabolcs~Iv\'an\inst{1}\thanks{This research was supported by the European Union and the State of Hungary, co-financed by the European Social Fund in the framework of T\'AMOP 4.2.4.A/2-11-1-2012-0001 ‘National Excellence Program’.}
\and \'Ad\'am~D.~Lelkes\inst{2}
\and Judit~Nagy-Gy\"orgy\inst{1}\thanks{Supported by the European Union and co-funded by the European Social Fund under the project ``Telemedicine-focused research activities on the field of Mathematics, Informatics and Medical sciences'' of project number ``T\'AMOP-4.2.2.A-11/1/KONV-2012-0073''}
\and Bal\'azs~Sz\"or\'enyi\inst{3,4}
\and Gy\"orgy~Tur\'an\inst{2,3}\thanks{Partially supported by NSF grant CCF-0916708}
}
\institute{
	University of Szeged, Hungary
	\and	
	University of Illinois at Chicago
	\and
	MTA-SZTE Research Group on Artificial Intelligence
	\and
	INRIA Lille, SequeL project, France
}

\def\Cov{{\mathrm{Cov}}}
\def\Rect{{\mathrm{Rect}}}
\def\Harm{{\mathrm{H}}}
\def\polylog{{\mathrm{polylog}}}
\def\bs#1{{\boldsymbol{#1}}}

\newtheorem{openproblem}{Problem}

\begin{document}

\maketitle

\begin{abstract}
We relate two complexity notions of bipartite graphs: the \emph{minimal weight biclique covering number} $\Cov(G)$ and
the \emph{minimal rectifier network size} $\Rect(G)$ of a bipartite graph $G$.
We show that there exist graphs with $\Cov(G)\geq \Rect(G)^{3/2-\epsilon}$.
As a corollary, we establish that there exist nondeterministic finite automata (NFAs) with $\varepsilon$-transitions,
having $n$ transitions total such that the smallest equivalent $\varepsilon$-free NFA has $\Omega(n^{3/2-\epsilon})$ transitions.
We also formulate a version of previous bounds for the weighted set cover problem and discuss its connections
to giving upper bounds for the possible blow-up.
\end{abstract}

\section{Introduction}
\label{sec-intro}
  In the world of descriptive complexity, questions involving the possible blow-up when transforming a description of some mathematical
  object from a formalism to another is a central topic, with one of the first papers dating back to 1971~\cite{conf/focs/MeyerF71}.
  We are primarily interested in the cost of chain rule removals from context-free grammars (CFGs).
  That is, how large a chain-rule free CFG has to be in the worst case which is equivalent to an input CFG of size $n$, having chain rules?
  The obvious upper bound resulting from the standard transformation is $O(n^2)$.
  The best known lower bound is $\Omega(n^{3/2-\epsilon})$~\cite{Blum1983287}.
  The question is interesting since chain rule elimination is the bottleneck part of the transformation to Chomsky Normal Form.
  Despite the question being well-motivated, we have no knowledge of progress in the last three decades; the gap is still there.

  The maximal possible blow-up is not known even in the special case of regular languages.
  When a \emph{regular} language is given (e.g. by a
  nondeterministic automaton or NFA, possibly having $\varepsilon$-transitions), an equivalent ``chain-rule-free'' \emph{regular} grammar
  corresponds to a nondeterministic automaton with no $\varepsilon$-transitions. In order to define the ``blow-up'', we
  have to choose a notion for measuring the \emph{size} of an NFA -- we say that the size of an NFA is the number of its \emph{transitions}.
  Regular languages can be represented by a variety of different formalisms, some of which are more concise than the others.
  For example, transforming a regular expression (RE) to an equivalent NFA can be done within linear bounds, i.e. the cost of
  this direction is worst-case $\Theta(n)$. From RE to $\varepsilon$-free NFA the worst-case cost is $\Theta(n\log^2 n)$,
  by the upper bound result of~\cite{Hromkovic:1997:TRE:646512.695338} and the matching lower bound of~\cite{Schnitger06regularexpressions}.
  The lower bound is achieved with a language possessing a linear-size RE as well, thus it is recognized by an NFA of size $O(n)$,
  hence the cost of the NFA $\to$ $\varepsilon$-free NFA transformation is $\Omega(n\log^2 n)$.
  However, the gap between $\Omega(n\log^2 n)$ and $O(n^2)$ has not been reduced since 2006.
  It is also known that from $\varepsilon$-free NFA to RE an exponential blow-up can occur and Kleene's algorithm produces an
  RE of exponential size from an NFA.

  One of the main results of the paper is that the NFA $\to$ $\varepsilon$-free NFA transformation has worst-case cost $\Omega(n^{3/2-\epsilon})$
  for any $\epsilon>0$. It is interesting that this bound (as well as the upper bound $O(n^2)$)
  coincides with that of~\cite{Blum1983287} for the seemingly more general problem of chain rule elimination.
  The methods (as well as the models) are very different but there is also a similarity: for the lower bound of~\cite{Blum1983287},
  languages consisting of words of length $3$ were defined. In our case, we consider languages consisting of words of length $2$.
  Such languages $L\subseteq\Sigma\Delta$ can be viewed as bipartite graphs $G_L=(\Sigma,\Delta,E_L)$ with $(a,b)$ being an edge in the graph
  iff the word $ab$ belongs to $L$. When the language is viewed this way, $\varepsilon$-free NFAs recognizing $L$
  correspond to biclique coverings~\cite{JuknaChapter2013} of $G_L$
  with the size of an NFA corresponding to the weight of the associated biclique covering.
  Also, NFAs recognizing $L$ correspond to rectifier networks~\cite{JuknaChapter2013} realizing $G_L$;
  again, with the size of an NFA corresponding to the size of the associated network.

  Hence, proving worst-case lower bounds for the minimum-weight covering of a bipartite graph having a rectifier network of size $n$,
  we get as byproduct worst-case lower bounds for the NFA $\to$ $\varepsilon$-free NFA transformation.
  Thus the bulk of the paper discusses biclique coverings and rectifier networks. These have also been studied for a long time in
  various contexts, see Sections~\ref{sec-notations} and \ref{sec-gap}.

  The paper is organized as follows. In Section~\ref{sec-notations} we give the notations we use for graphs and automata.
  In Section~\ref{sec-gap} 
  we give lower bounds for the possible blow-up between rectifier network size and biclique covering weight.
  In Section~\ref{sec-setcover} we give upper bounds for this blow-up and consider
  the biclique covering problem as a weighted set cover problem. An approximation bound for the greedy algorithm
  given by Lov\'asz~\cite{Lovasz1975383} for the unweighted case is generalized to the weighted case.
  We discuss the connection of this bound to possible upper bounds for the blow-up.
  In Section~\ref{sec-app} we relate these graph-theoretic results to automata theory and prove the aforementioned
  lower bound of $\Omega(n^{3/2-\epsilon})$ for $\varepsilon$-removal.
\section{Notations}
\label{sec-notations}
\subsubsection{Graphs, biclique coverings and rectifier networks}
  Let $[n]$ stand for the set $\{1,\ldots,n\}$.
  For sets $A$ and $B$, $K_{A,B}$ stands for the complete bipartite graph $(A,B,A\times B)$.
  When only the cardinalities $a$ and $b$ of the sets $A$ and $B$ matter, we write $K_{a,b}$ for $K_{A,B}$.
  When $G=(A,B,E)$ is a bipartite graph,
  a \emph{biclique} of $G$ is a complete bipartite subgraph of $G$ and
  the \emph{weight of a biclique} is the number of its vertices.
  A \emph{biclique covering} of $G$ is a collection $\mathcal{C}$ of its bicliques such that each edge of $G$ belongs to at least one member of $\mathcal{C}$,
  the \emph{weight of a covering} is the sum of the weights of the bicliques present in the covering
  and $\Cov(G)$ is the minimum possible weight of a biclique covering of $G$.

  A biclique $K_{a,b}$ has weight $a+b$ while it covers $ab$ edges of $G$.
  In our investigations we will frequently use the inverse $\frac{ab}{a+b}$ of the relative cost of covering the edges by $K_{a,b}$.
  We introduce the shorthand $\Harm(a,b)$ to denote the quantity $\frac{ab}{a+b}$.

  For a bipartite graph $G=(A,B,E)$, a \emph{rectifier network realizing $G$} is a directed
  acyclic graph (DAG) $R=(V,E')$ with $A$ being the set of source nodes of $R$ and $B$ being the set of sink nodes of $R$, satisfying the property
  that $(a,b)\in E$ if and only if $b$ is reachable from $a$ in $R$. The \emph{size} of a rectifier network is the number of its edges.
  The \emph{depth} of a network is the length of its longest path. We let $\Rect(G)$ stand for the size of the smallest rectifier network
  realizing $G$ and $\Rect_k(G)$ for the size of the smallest rectifier network of depth at most $k$ realizing $G$.
  We may assume \emph{w.l.o.g.} that there are no isolated vertices.
\newcommand{\bipgraph}[2]{%
    \begin{tikzpicture}[every node/.style={circle,draw}, scale=0.6]
    \foreach \xitem in {1,2,...,#1}
    {%
    \node at (0,\xitem) (a\xitem) {};
    \node at (2,\xitem) (b\xitem) {};
    }%

    \foreach \x [count=\xi] in {#2}
    {%
    \foreach \tritem in \x
    \draw(a\xi) -- (b\tritem);
    }
    \end{tikzpicture}
}
  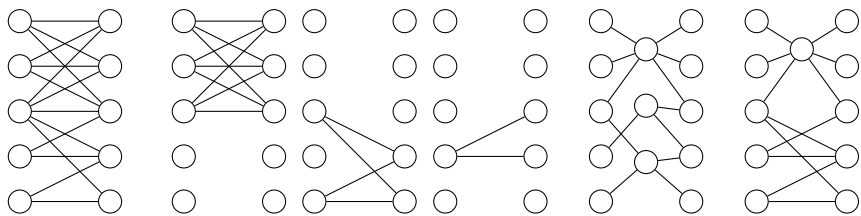
\begin{figure}[ht!]\centering
  \bipgraph{5}{{1,2},{2,3},{1,2,3,4,5},{3,4,5},{3,4,5}} \hfil
  \bipgraph{5}{{},{},{3,4,5},{3,4,5},{3,4,5}}
  \bipgraph{5}{{1,2},{},{1,2},{},{}}
  \bipgraph{5}{{},{2,3},{},{},{}}\hfil
    \begin{tikzpicture}[every node/.style={circle,draw}, scale=0.6]
    \foreach \xitem in {1,2,...,5}
    {%
    \node at (0,\xitem) (a\xitem) {};
    \node at (2,\xitem) (c\xitem) {};
    }%
    \node at (1,1.5*5.0/4.0) (b1) {};
    \node at (1,2.5*5.0/4.0) (b2) {};
    \node at (1,3.5*5.0/4.0) (b3) {};
    \draw (a1)--(b1)--(c1);
    \draw (a3)--(b1)--(c2);
    \draw (a2)--(b2)--(c2);
    \draw (b2)--(c3);
    \draw (a3)--(b3)--(c3);
    \draw (a4)--(b3)--(c4);
    \draw (a5)--(b3)--(c5);
    \end{tikzpicture}  \hfil
    \begin{tikzpicture}[every node/.style={circle,draw}, scale=0.6]
    \foreach \xitem in {1,2,...,5}
    {%
    \node at (0,\xitem) (a\xitem) {};
    \node at (2,\xitem) (c\xitem) {};
    }%
    \node at (1,3.5*5.0/4.0) (b3) {};
    \draw (a1)--(c1)--(a3);
    \draw (a1)--(c2)--(a3);
    \draw (c3)--(a2)--(c2);
    \draw (a3)--(b3)--(c3);
    \draw (a4)--(b3)--(c4);
    \draw (a5)--(b3)--(c5);
    \end{tikzpicture}
  \caption{From left to right: a graph $G$, three bicliques showing $\Cov(G)\leq 13$, a depth-$2$ network corresponding to the bicliques
  having size $13$, and another network showing $\Rect_2(G),\Rect(G)\leq 12$. In the networks, edges are directed from left to right.}
  \end{figure}

  There are constructions of graphs for which only large rectifier networks exist (i.e. having large $\Rect$ value), the dates of the results ranging
  from 1956 till 1996, e.g. graphs $G$ on $n$ vertices with
  $\Rect(G)$ being $\Omega(n^{3/2})$~\cite{Nechiporuk}, $\Omega(n^{5/3})$~\cite{Mehlhorn,Pippenger:1976:SGA:321958.321962,Wegener}
  and $\Omega(n^{2-\epsilon})$~\cite{kollar1996norm}. Also, it is known that $\Rect(G)\leq \frac{n^2}{\log n}$~\cite{Lupanov}.

  In this paper we are interested in the largest possible gap between $\Cov$ and $\Rect$, thus we seek graph classes having a
  \emph{small} $\Rect$ and a \emph{large} $\Cov$ value.

  For $\Cov$, a related notion is that of Steiner $2$-transitive-closure\--spanners~\cite{BermanEtal} (Steiner-$2$-TC-Spanners),
  which is a more general notion for realizing general graphs. The two notions coincide when we look for
  spanners of bipartite graphs, viewed as $2$-level layered directed graphs.
  The authors of~\cite{BermanEtal} show a lower bound for the minimal Steiner-$2$-TC-Spanner a bipartite graph can have.
  Applying these results to our problem, we get that there exist graphs with $\Rect(G)=O(n)$ and $\Cov(G)=\Omega(n\mathrm{polylog}(n))$
  which is exactly the type of result we seek to achieve.
  We use the asymptotic behaviour operators $O$, $\Omega$ and $\Theta$ as well as their ``up to a polylogarithmic factor'' variants
  $\tilde{O}$, $\tilde{\Theta}$, e.g. $f(n)=\tilde{O}(g(n))$ is a shorthand for ``$f=O(g(n)\log^kg(n))$ for some constant $k\geq 0$''.
\subsubsection{Automata}
  A nondeterministic finite automaton, or NFA for short, is a tuple $M=(Q,\Sigma,\delta,q_0,F)$ with $Q$ being an alphabet of states,
  $\Sigma$ being the input alphabet, $\delta\subseteq Q\times\Sigma_\varepsilon\times Q$ a transition relation where
  $\Sigma_\varepsilon$ denotes the set $\Sigma\cup\{\varepsilon\}$, $q_0\in Q$ being the start state and $F\subseteq Q$ being
  the set of accepting states. The automaton is $\varepsilon$-free if there is no transition of the form $(p,\varepsilon,q)\in\delta$.

  A \emph{run} of the above $M$ is a sequence $(p_1,a_1,r_1)\ldots(p_t,a_t,r_t)\in \delta^*$ such that for each
  $1\leq i<t$, $r_{i}=p_{i+1}$, and $p_1=q_0$. The run is accepting if $r_t\in F$. The label of the run is the $\Sigma$-word
  $a_0a_1\ldots a_t$. The \emph{language} recognized by $M$ is $L(M)=\{w\in \Sigma^*:\textrm{there is an accepting run of }M\textrm{ with label }w\}$.

  The \emph{size} of an NFA $M$ is the cardinality $|M|$ of its set $\delta$ of transitions.
  It is well-known that for each NFA $M$ there exists an equivalent $\varepsilon$-free automaton $M'$ with $|M'|=\mathcal{O}(|M|^2)$, i.e.
  $\varepsilon$-elimination can be achieved via a quadratic blow-up. However, no explicit lower bounds are stated in the literature.
\section{Lower bounds for the blow-up}
\label{sec-gap}
  It is clear that $\Rect_{k+1}(G)\leq \Rect_k(G)$ for each $k\geq 0$, and that there exists some $k\geq 0$ with $\Rect_k(G)=\Rect(G)$ and
  $\Rect_k(G)=\Rect_{k'}(G)$ for every $k'>k$. Moreover, $\Rect_2(G)\leq \Cov(G)\leq 2\cdot\Rect_2(G)$: for any collection $\mathcal{C}$ of bicliques
  one can construct a rectifier network $R=(A\uplus\mathcal{C}\uplus B,E')$ with $(a,K_{A',B'})$ and $(K_{A',B'},b)$ being an edge iff
  $a\in A'$ and $b\in B'$, respectively, showing $\Rect_2(G)\leq \Cov(G)$. For $\Cov(G)\leq 2\cdot\Rect_2(G)$, let
  $R=(A\uplus X\uplus B,E')$ be a depth-$2$ rectifier network realizing $G$.
  Then, edges of $E'$ are directed from $A$ to $X$, from $X$ to $B$ and also ``jump edges'' from $A$ directly to $B$ are allowed.
  First, subdividing each such jump edge and adding the intermediate node to $X$ eliminates jump edges and the resulting network
  $R'=(A\uplus X'\uplus B,E'')$ still realizes $G$ in depth $2$ and due to the subdividing, $|E''|\leq 2\cdot|E'|$.
  For a node $x\in X'$, let $A(x)$ be the set of its ancestors (in $A$) and $B(x)$ be the set of its descendants (in $B$).
  Note that if $R$ is minimal, then neither of these sets is empty.
  Then in $G$, each member of $B(x)$ is reachable from any member of $A(x)$, hence $K_{A(x),B(x)}$ is a biclique of $G$ and the
  collection $\mathcal{C}=\{K_{A(x),B(x)}:x\in X'\}$ is a biclique cover of $G$ of size $|E''|\leq 2\cdot|E'|=2\cdot\Rect_2(G)$.
  Observe that the factor of $2$ is tight e.g. in the case of complete matchings.

  Since adding or removing isolated nodes to $G$ does not affect either $\Cov(G)$ or $\Rect(G)$, from now on we assume that $G$ has no isolated vertices.

  It is also clear that \[n\leq\Rect(G)\leq\Cov(G)\leq 2|E(G)|\] where $n$ stands for the number of vertices\footnote{At times $n$ will denote the size of \emph{one} of the two classes of $G$, introducting a factor of $2$ but never causing differences in the growth order.} of $G$:
  in any rectifier network the outdegree of each node $a\in A$ is at least one, and the collection $\{K_{\{a\},\{b\}}:(a,b)\in E\}$ of bicliques
  is a covering of weight $2|E(G)|$. Hence, $\Cov(G)=O(\Rect^2(G))$.
  However, it is not known whether the quadratic gap is attainable: in the rest of the article we seek an $\alpha>1$, being as high as possible,
  such that there exist graphs with arbitrary large $\Rect(G)$ and with $\Cov(G)=\Omega(\Rect^\alpha(G))$.

  To this end, we have to construct graph families having \emph{small} $\Rect$ and \emph{large} $\Cov$.
  To show $\Rect$ is small (usually it will be $\tilde{O}(n)$ in our candidates) it suffices to give a small realizing network.
  On the other side, to see that $\Cov$ is large, we should have good lower bound methods.

  For providing lower bounds, we define the following parameter $\kappa(G)$ of a bipartite graph $G$: let
  \begin{align*}
  \kappa(G)\quad&=\quad\max\{\Harm(|A'|,|B'|):K_{A',B'}\textrm{ is a biclique of }G\}
  \end{align*}
  Observe that by monotonicity, it suffices to take \emph{maximal} bicliques of $G$ into account.

  This graph parameter provides lower bounds not only for $\Cov(G)$ but for $\Rect(G)$:
  \begin{proposition}[See e.g.~\cite{JuknaChapter2013}, Lemma 1.10. and Theorem 1.72.]
  \label{prop-kappa}
  For any bipartite graph $G=(A,B,E)$, it holds that $\frac{|E|}{\kappa(G)}\leq\Cov(G)$ and $\frac{|E|}{\kappa(G)^2}\le\Rect(G)$.
  \end{proposition}
  By a similar argument, we can obtain the following inequality as well:
  \begin{proposition}
  \label{prop-cov-rect2kappa}
  For any bipartite graph $G$, it holds that $\Cov(G)\leq\Rect(G)\cdot 2\kappa(G)$.
  \end{proposition}
  \begin{proof}
  Claim 1.73. in~\cite{JuknaChapter2013} states the following. Let $k$ be the maximum integer with $K_{k,k}$ being
  a biclique of $G$. For any rectifier network $R=(V,E')$ realizing $G$, call an edge $(u,v)\in E'$
  \emph{eligible} iff $|A(u)|\leq k$ and $|B(v)|\leq k$. Then for any edge $(a,b)\in E$ there is a path from $a$ to $b$
  in $R$ containing an eligible edge.

  In that case $\{K_{A(u),B(v)}:(u,v)\in E'\textrm{ is eligible}\}$
  is a covering of $G$, consisting of at most $|E'|=\Rect(G)$ bicliques.
  Each biclique has weight at most $2k$ which in turn is
  at most $2\kappa(G)$ since $H(a,b)\leq \min\{a,b\}$ holds for any $a,b>0$.
  \end{proof}

  It is also worth observing that $k=\Theta(\kappa)$
  since $\min\{a,b\}\leq 2H(a,b)$.

  Our first result considers the bipartite graph corresponding to the $\mod 2$ inner product function.
  \begin{theorem}
  \label{thm-orthogonal}
  Let $d>0$ be an even integer and $G^d_\bot=(A,B,E)$ be the bipartite graph with $A=B=\{0,1\}^d$ and
  $(\bs{u},\bs{v})\in E$ for the vectors $\bs{u},\bs{v}\in\{0,1\}^d$ iff $\bs{u}\bot\bs{v}$ in $\mathbb{Z}_2^d$, i.e.
  iff $\sum\limits_{i\in[d]}u_iv_i=0$ where sum is taken modulo $2$.

  Then $\Rect(G^d_\bot)=\tilde{O}(n)$ and $\Cov(G^d_\bot)=\Omega(n^{3/2})$
  where $n=2^d$ is the number of vertices of $G^d_\bot$.
  \end{theorem}

  The proof is broken into two parts. 
  The lower bound follows from the first inequality of Proposition~\ref{prop-kappa}
  and a special case of Lindsey's lemma~\cite{DBLP:books/daglib/0028687}.
  \begin{proposition}
  \label{prop-bot-cov}
  $\kappa(G^d_\bot)=\frac{\sqrt{n}}{2}$. Thus $\Cov(G^d_\bot)=\Omega(n^{3/2})$.
  \end{proposition}
  At the same time, $\Rect(G^d_\bot)$ is small enough.
%
%
  To see this, we show $\Rect(G)=\tilde{O}(n)$ for a specific family of bipartite graphs, which we call \emph{permutation invariant} graphs.
  A bipartite graph $G=(\{0,1\}^d,\{0,1\}^d,E)$ is permutation invariant if $(\bs{u},\bs{v})\in E$ implies $(\pi(\bs{u}),\pi(\bs{v}))\in E$
  for any permutation $\pi:[d]\to[d]$ of the coordinate index set. Here $\pi(u_1,\ldots,u_d)$ is defined to be $(u_{\pi(1)},\ldots,u_{\pi(d)})$.
  It is clear that the graphs $G_\bot^d$ are permutation invariant.

  For such graphs the following holds (which also state that within this class of graphs, the bound $\Cov(G)=\Omega(\Rect(G)^{3/2})$ is optimal):
  \begin{theorem}
  \label{thm-perm}
  For permutation invariant graphs $\Rect$ is $\tilde{O}(n)$ and $\Cov$ is $\tilde{O}(n^{3/2})$.
  \end{theorem}
  \begin{proof}
  Suppose $G=(A,B,E)$ is permutation invariant with $A=B=\{0,1\}^d$. Let $c:\{0,1\}^d\times\{0,1\}^d\ \to\ \{0,\ldots, d\}^{\{0,1\}\times\{0,1\}}$ be the function defined as
  \[c((u_1,\ldots,u_d),(v_1,\ldots,v_d))(a,b)\ =\ |\{i\in[d]:u_i=a,v_i=b\}|.\]
  That is,
  $c(\bs{u},\bs{v})(a,b)$ is the number of positions $i$ on which $\bs{u}$ is $a$ and $\bs{v}$ is $b$.

  Then, $G$ factors through $c$ in the following sense: if $c(\bs{u},\bs{v})=c(\bs{u}',\bs{v}')$, then $(\bs{u},\bs{v})\in E$ iff $(\bs{u}',\bs{v}')\in E$.
  Indeed, $c(\bs{u},\bs{v})=c(\bs{u}',\bs{v'})$ if and only if there exists a permutation $\pi:[d]\to[d]$ such that $u_i=u'_{\pi(i)}$ and
  $v_i=v'_{\pi(i)}$ for each $i\in [d]$, yielding $(\bs{u},\bs{v})\in E$ if and only if $(\bs{u}',\bs{v}')\in E$.

  Hence there exists a subset $C$ of the finite set $\{0,\ldots,d\}^{\{0,1\}\times\{0,1\}}$ such that $(\bs{u},\bs{v})\in E$ iff $c(\bs{u},\bs{v})\in C$.

  We define a rectifier network $R=(\{0,1\}^d\times \{0,\ldots,d\}^{\{0,1\}\times\{0,1\}} \times\{0,\ldots,d\})$:
  the pair $((u_1,\ldots,u_d),f,\ell),((v_1,\ldots,v_d),f',\ell')$ is an edge of $R$ iff the following conditions hold:
  $\ell'=\ell+1$ (so that $R$ is a DAG of depth $d+1$); for each $i\neq\ell'$, $u_i=v_i$ holds; finally,
  $f'(u_{\ell'},v_{\ell'})=f(u_{\ell'},v_{\ell'})+1$ and for any other $(a,b)\in\{0,1\}\times\{0,1\}$, $f'(a,b)=f(a,b)$.
  
  Then by induction on $\ell'-\ell$ 
  we get that
  there is a path from $((u_1,\ldots,u_d),f,\ell)$ to
  $((v_1,\ldots,v_d),f',\ell')$ iff the following conditions hold:
  $\ell < \ell'$; for each $i\leq\ell$ and $i>\ell'$, $u_i=v_i$; finally, $f'(a,b)=f(a,b)+|\{\ell<i\leq\ell':u_i=a,v_i=b\}|$.
  
  Now let $R'=(V(R)\uplus \{0,1\}^d,E(R)\uplus E')$ with $E'$ consisting of the edges of the form $(\bs{v},f,d)\to \bs{v}$ with $f\in C$.
  Then $R'$ realizes $G$ by identifying each $\bs{u}\in A$ with $(\bs{u},\bs{0},0)$ and each $\bs{v}\in B$ with the element $\bs{v}$ of this
  last layer of $R'$ (here $\bs{0}$ stands for the constant zero function $\bs{0}:(a,b)\mapsto 0$).
  Since in $R'$, there are at most $2^d(\cdot d\cdot\{d+1\}^4)\cdot 2$ edges (each node not belonging to layer $d$
  has outdegree $2$ in $R$ and in the last step, $2^d\times |C|\leq 2^d\cdot \{d+1\}^4$ edges are added), which is
  $O(n\log^5n)$, showing $\Rect(G)=\tilde{O}(n)$.

  For $\Cov(G)=\tilde{O}(n^{3/2})$, let $X$ be the set of vertices of $R'$ of the form $(\bs{u},f,d/2)$. (That is, nodes of the middle layer of $R'$.)
  As before, let $A(x)\subseteq A$, $x\in X$ stand for the set of nodes from which $x$ is reachable in $R'$ and let $B(x)\subseteq B$ stand
  for the set of those nodes which are reachable from $x$ in $R'$.
  Then, since each node of $R$ has indegree at most $2$, we have that $|A(x)|\leq 2^{d/2}$. For $|B(x)|$, since each outdegree in $R$ is $2$,
  we get that there are at most $2^{d/2}$ nodes of the form $(\bs{v},f,d)$ reachable from $x$. In $E'$, the
  outdegree of these nodes is $|C|$ which is at most $(d+1)^{4}$, hence $|B(x)|\leq 2^{d/2}(d+1)^4=\tilde{O}(\sqrt{n})$.
  Thus, the covering $\mathcal{C}=\{K_{A(x),B(x)}:x\in X\}$ has size $\sum_{x\in X}(|A(x)|+|B(x)|)$ which is at most
  $2^d\cdot (d+1)^4\cdot(2^{d/2}+2^{d/2}(d+1)^4)=\tilde{O}(n^{3/2})$. Note that due to the layered structure of $R'$, each $\bs{u}\to\bs{v}$ path
  contains a node belonging to $X$, so $\mathcal{C}$ is indeed a covering.
  \end{proof}

  Thus we have showed that for an arbitrarily $\epsilon>0$ there are graphs $G=G_\bot^d$ having arbitrarily large $\Rect(G)=\tilde{O}(n)$ and with $\Cov(G)=\Omega(\Rect^{3/2-\epsilon}(G))$.
  (Observe that any permutation invariant graph with $\kappa=\Theta(\sqrt{n})$ and $\Theta(n^2)$ edges meets this condition.)

  As an interesting corollary, we get that $\Cov(G_\bot^d)$ is $\tilde{\Theta}(n^{3/2})$ which is $\tilde{\Theta}(\frac{|E|}{\kappa})$
  so in this case the bound of Proposition~\ref{prop-kappa} is optimal up to a log factor.

  A general construction for constructing
  a biclique covering of a graph $G=(A,B,E)$ is the following: starting from a rectifier network $R=(V,E')$
  first one chooses a \emph{cut} $E_0\subseteq E'$ of the edges of $R$ (so that each $a\to b$, $a\in A$, $b\in B$ path contains an edge
  from $E_0$), in which case a covering is $\mathcal{C}(E_0)=\{K_{A(x),B(y)}:(x,y)\in E_0\}$.
  We call coverings of this form \emph{cut-coverings of $R$}.
  (In the proof of Theorem~\ref{thm-perm} we employ a similar construction, choosing a subset $X$ of vertices instead of a subset $E_0$.)
  In the following we state without proof that this
  construction is not optimal, not even up to a polylogarithmic factor, even when $R$ is optimal up to a polylogarithmic factor.
  \begin{theorem}
  \label{thm-negyed}
  Consider the graph $G_\Delta^n=(A,B,E)$ with $A=B=[n]$ and $(i,j)\in E$ iff $d(i,j)\leq\frac{n}{4}$ where $d(i,j)$
  is the modulo $n$ distance $\min\{|i-j|,|n+i-j|\}$. (That is, distance on the circle graph $C_n$.) Then:
  \begin{enumerate}
  \item There exists a rectifier network $R_n$ realizing $G_\Delta^n$ with $\tilde{O}(n)$ edges.
  \item Any cut-covering of $R_n$ has size $\Omega(n^2)$.
  \item At the same time, $\Cov(G_\Delta^n)$ is $O(n^{1+\epsilon})$ for any $\epsilon>0$ where the $O$ notation hides a constant depending only on $\epsilon$.
  \end{enumerate}
  \end{theorem}
  Note that for this graph we have $\kappa=\Theta(n)$ since $K_{[n/4],[n/4]}$ is a biclique.
  Hence also for this class of graphs, $\frac{|E|}{\kappa}=\Theta(n)$ approximates $\Cov=O(n^{1+\epsilon})$ relatively well.
  In the next section we show that a closely related formula gives an upper bound for $\Cov(G)$.
\section{Upper bounds for the blow-up}
\label{sec-setcover}
In this section we will show that, under certain assumptions, $\Cov(G)=o(\Rect(G)^2)$ or even $\Cov(G) = O(\Rect(G)^{3/2})$ holds.
Proposition~\ref{prop-cov-rect2kappa} implies the following result:
\begin{theorem} \label{thm:feltetel}
  For any bipartite graph $G$ and $0<\alpha\le1$ with $\Cov(G)\leq\frac{|E|}{\kappa^\alpha}$ we have
  $\Cov(G)\leq 2\Rect(G)^{\beta}$ for some $\beta\leq 1+\frac{1}{1+\alpha}\in[3/2,2)$.
  Hence if $\Cov(G)\leq \frac{|E|}{\kappa^\alpha}$ holds for a family of graphs $G$, then
  $\Cov(G)=O(\Rect(G)^{2-\frac{\alpha}{1+\alpha}})$.
\end{theorem}
  \begin{proof}
  Let us introduce the following notation: $|E|=n^\delta$ for $n=|V(G)|$, $\Rect(G)=n^r$ and $\Cov(G)=\frac{|E|}{\kappa^\alpha}$, $0<\alpha\le1$.
  We will show that 
  choosing $\beta=\frac{\delta+\alpha\cdot r}{r(1+\alpha)}$ suffices. (Note that since $\delta\leq 2$ and $r\geq 1$, $\beta$ is indeed at most 
  $1+\frac{1}{1+\alpha}$.)

  By $\Cov=\frac{|E|}{\kappa^\alpha}$ we have $\log_n\Cov = \delta-k\alpha$ where $k=\log_n\kappa$.
  Now assuming for contradiction that $2^{\frac{\alpha}{1+\alpha}}\Rect^\beta<2\Rect^\beta<\Cov$ we get
  \[r\beta=\frac{\delta+\alpha\cdot r}{1+\alpha}<\delta-k\cdot \alpha-\frac{\alpha}{1+\alpha}\log_n 2.\]
  Then direct computation shows that
  $r<\delta-k(1+\alpha)-\log_n 2$
  which is a contradiction, since by $\Rect\geq \frac{\Cov}{2\kappa}$ we have $r\geq \delta-k(1+\alpha)-\log_n 2$.
  \end{proof}
  Simple examples show that the assumption of the theorem does not hold for all graphs.
  A similar argument gives a similar, but somewhat weaker, bound $\Cov(G)=O(\Rect(G)^{2-\varepsilon})$ for some $\varepsilon>0$
  if the condition $\Cov(G)\leq\frac{|E|}{\kappa^\alpha}$ is replaced by $\Cov(G)\leq \polylog n \max\frac{|E(G')|}{\kappa(G')}$,
  where the maximum ranges over induced subgraphs $G'$ of $G$.
  Thus an affirmative answer to the following open problem would imply $\Cov(G)=O(\Rect(G)^{2-\varepsilon})$ for all bipartite graphs.
  \begin{openproblem}
  Is it true that for any bipartite graph $G$ on $n$ vertices,
  \[\Cov(G)\leq \polylog n \max\frac{|E(G')|}{\kappa(G')}\]
  where the maximum ranges over induced subgraphs $G'$ of $G$?
  \end{openproblem}
\subsection{The set cover problem}
  Now we 
  apply the weighted set cover problem to our setting.
  For
  a
  detailed discussion of this problem, and 
  an introduction
  to approximation 
  methods
  see~\cite{Williamson:2011:DAA:1971947}.

  The \emph{weighted set cover} problem is the following: we are given a collection $\mathcal{S}=\{S_1,\ldots,S_t\}$ of subsets of some finite
  universe $A$ of $n$ elements with $\cup\mathcal{S}=A$, and to each $S_i$, a cost $c(S_i)>0$ is associated.
  The goal is to find a subset $\mathcal C$ of $\mathcal S$ such that $\cup\mathcal C=A$ and the total cost $\sum_{S\in\mathcal C} c(S)$ is minimized.
  The problem is well-known to be NP-complete already for the uniform setting when $c(S_i)=1$; however, the following greedy algorithm
  returns a fair enough approximation:
  \begin{algorithmic}
  \State Let $U:=A$ and $\mathcal{C}:=\emptyset$.
  \While{$U\neq\emptyset$}
  \State Choose $S\in\mathcal{S}$ such that $\frac{c(S)}{|S\cap U|}$ is the minimum possible value.
  \State Let $U:=U-S$ and $\mathcal{C}:=\mathcal{C}\cup\{S\}$.
  \EndWhile
  \Return $\mathcal{C}$.
  \end{algorithmic}
%
%
   The following linear program is the standard relaxation of weighted set cover:
\begin{align*}
  \textrm{minimize }  \sum_{i=1}^t c(S&_i) x_i &
  \textrm{subject to } \sum_{i: a\in S_i} x_i &\ge 1 \quad \forall a\in A,\quad x_i\ge 0
  \end{align*}
Denote by $\mathrm{OPT}$ the optimal solution of the weighted set cover problem.
It is well-known~\cite{chvatal1979setcover,Johnson:1973:AAC:800125.804034,Lovasz1975383,Stein1974391} that the value of the solution returned by the above algorithm is 
  bounded by $\ln n \cdot \mathrm{OPT}$, where $n=|A|$, and even by $\ln n \cdot Z^*_{LP}$,
  where $Z^*_{LP}$ denotes the value of an optimal solution to the LP relaxation.

  Now we define a related combinatorial quantity.
  For a subset $B$ of $A$, let $\eta(B)$ stand for the value $\mathop{\min}_{S\in\mathcal{S}}\frac{c(S)}{|S\cap B|}$ which is
  present inside the loop of the greedy algorithm. Note that this value is positive and finite for any $B\subseteq A$.
  Also, let $\eta^*$ stand for $\mathop{\max}_{B\subseteq A}|B|\cdot \eta(B)$. Then we have:
  \begin{proposition}
  \label{prop-rhostar-opt}
  $\eta^*\leq Z^*_{LP}\leq\mathrm{OPT}$.
  \end{proposition}
  \begin{proof}
  Consider any feasible solution $\bs{x}$ and a subset $B$ of $A$.
  Then,
  \begin{alignat*}{3}
  \sum_{i\in [t]}c(S_i)x_i
  &\geq \sum_{i:S_i\cap B\neq \emptyset}c(S_i)x_i
  = \sum_{i:S_i\cap B\neq\emptyset}\sum_{a\in S_i\cap B}\tfrac{c(S_i)}{|S_i\cap B|}x_i\\
  &\geq \sum_{i:S_i\cap B\neq\emptyset}\sum_{a\in S_i\cap B}\eta(B)x_i
  = \eta(B)\sum_{a\in B}\sum_{S_i\ni a}x_i\\
  &\geq \eta(B)\sum_{a\in B}1
  = \eta(B)\cdot|B|.
  \end{alignat*}
  \end{proof}

  On the other hand, this quantity can be used to give an upper bound for OPT as well.
  The following bound is proven in~\cite{LovaszThesis} for the unweighted case.
  \begin{proposition}
  \label{prop-greedy}
  Let $\textsc{Greedy}$ stand for the cost of the solution returned by the greedy algorithm. Then $\textsc{Greedy}\leq H_n\eta^*$,
  where $H_n$ is $\sum_{i\in[n]}\frac{1}{i}\approx \ln n$.
  \end{proposition}
  \begin{proof}
  Let $U_k$ denote the set of uncovered elements at the beginning of the $k$th iteration of the loop of the greedy algorithm
  and let $n_k=|U_k|$.
  Then
  \begin{alignat*}{3}
  \min_{S_i}\tfrac{c(S_i)}{|S_i\cap U_k|}
  &=\min_{S_i}\tfrac{c(S_i)}{|S_i\cap U_k|}\tfrac{|U_k|}{n_k}
  \leq \max_{B\subseteq A}\min_{S_i}\tfrac{c(S_i)}{|S_i\cap B|}\tfrac{|B|}{n_k}\\
  &= \max_{B\subseteq A}\tfrac{|B|}{n_k}\min_{S_i}\tfrac{c(S_i)}{|S_i\cap B|}
  = \max_{B\subseteq A}\tfrac{|B|}{n_k}\eta(B)
  = \tfrac{1}{n_k}\eta^*.
  \end{alignat*}
  Thus, the covering of the elements covered in the $k$th iteration costs at most $\frac{1}{n_k}\eta^*$ for each such element.
  Since $n_1=|A|=n$ and at each iteration, $n_k$ is strictly decreasing,
  the total cost is at most $\sum_{i\in[n]}\frac{1}{i}\eta^*=H_n\eta^*$.
  \end{proof}
  Thus we have the following chain of inequalities:
  \[ \eta^*\le Z^*_{LP}\le\textsc{OPT}\leq\textsc{Greedy}\le H_n\cdot \eta^*. \]
%
%
%
\subsection{Application to biclique coverings}
\label{sec-setcover-cov}
  Determining $\Cov(G)$ for $G=(A,B,E)$ can be viewed as a set cover problem: the universe is $E$,
  the allowed sets are bicliques of $G$, and the cost of a biclique $K_{A',B'}$ is $|A'|+|B'|$.

  In this problem, $\eta$ is the following:
  given a subset $E'$ of $E$ (that is, a \emph{subgraph} $G'=(A',B',E')$ of $G$),
  $\eta(E')$ is defined as
  \[\mathop{\min}_{K_{A'',B''}\subseteq E}\tfrac{|A''|+|B''|}{|E'\cap (A''\times B'')|}\]
  and $\eta^*$ is the maximal possible value of $|E'|\eta(E')$.
  By Proposition~\ref{prop-greedy} we have that $\Cov(G)\leq\eta^*\cdot H_n$.

  Observe that for the biclique $K_{A_0,B_0}=\mathrm{arg\,min}_{K_{A',B'}\subseteq E}\frac{|A'|+|B'|}{|E'\cap (A'\times B')|}$
  we have $A_0\subseteq A'$ and $B_0\subseteq B'$. Indeed, otherwise $K_{A_0\cap A',B_0\cap B'}$ would be a better biclique.
  Thus, the minimizer biclique $K_{A_0,B_0}$ is a biclique of the subgraph of $G$ \emph{induced} by $A_0\cup B_0$.
  It is also clear that for induced subgraphs $G'$ we have $\frac{1}{\kappa(G')}=\eta(E')$ and $|E'|\eta(E')=\frac{|E(G')|}{\kappa(G')}$.

  Hence, there are two cases:
  either $\eta^*$ takes its value on some \emph{induced} subgraph of $G$ up to a polylogarithmic factor,
  in which case the bound $\Cov(G)\geq \max\frac{|E(G')|}{\kappa(G')}$ is essentially optimal up to a polylog factor,
  or not, in which case there are graphs having much larger $\eta^*$ than $\frac{|E|}{\kappa}$.
  In the first case, the remarks following Theorem~\ref{thm:feltetel} imply a subquadratic upper bound for the blow-up.

  \begin{openproblem}
  Determine the gap possible between $\eta^*$ and $\max\frac{|E(G')|}{\kappa(G')}$, where the maximum is taken over all induced subgraphs.
  \end{openproblem}

\section{Application: the cost of $\varepsilon$-removal}
\label{sec-app}
Let $A$ and $B$ be disjoint alphabets (nonempty finite sets) and $L\subseteq AB$ a (finite) language consisting of two-letter words.
Then $L$ can also be viewed as a bipartite graph $G_L=(A,B,L)$ where the notation for $L$ is slightly abused (i.e. $(a,b)$ is an edge
iff the word $ab$ belongs to the language). Without loss of generality we may assume that for each $a\in A$ ($b\in B$, resp.)
there exists a $b\in B$ ($a\in A$, resp.) such that $ab$ is in $L$.

\begin{proposition}
There is some NFA $M$ recognizing $L$ with $|M|=O(\Rect(G_L))$.
\end{proposition}
\begin{proof}
Let $R=(V,E)$ be a rectifier network for $G_L$
with $|E|=\Rect(G_L)$.
Then the automaton $M=(V\uplus\{q_0,q_f\},A\cup B,\delta,q_0,\{q_f\})$ with
\[\delta=\{(q_0,a,a):a\in A\}\cup\{(b,b,q_f):b\in B\}\cup\{(p,\varepsilon,q):p\to q\in E\}\]
recognizes $L$ with $|M|=|E|+|A|+|B| = O(\Rect(G_L))$.
\end{proof}

\begin{proposition}
For any $\varepsilon$-free NFA $M$ recognizing $L$, $\Cov(G_L)\leq |M|$.
Moreover, there exists a $\varepsilon$-free NFA $M$ recognizing $L$ with $|M|=\Cov(G_L)$.
\end{proposition}
\begin{proof}
Let $M=(Q,A\cup B,\delta,q_0,F)$ be an $\varepsilon$-free NFA recognizing $L$ of minimal size.
Since $L$ is prefix-free and $M$ is minimal, $F=\{q_f\}$ is a singleton set.
Also, $M$ is \emph{trim}, i.e. for each state $p\in Q$ there exist words $x,y$ with $p\in q_0x$ and $q_f\in py$.
Since every word in $L$ has the same length $2$, to each state $p$ there is an integer $0\leq n_p\leq 2$ such that
whenever $p\in q_0x$ for some word $x$, then $|x|=n_p$. Otherwise if $p\in q_0x_1$ and $p\in q_0x_2$ for words $x_1,x_2$
of different length, then $x_1y$ and $x_2y$ are members of $L$ of different length for any word $y$ with $q_f\in py$,
a contradiction.
Also, it is clear that $n_p=0$ only for $p=q_0$ and $n_p=2$ only for $p=q_f$. Thus if $X$ stands for $Q-\{q_0,q_f\}$,
we get that $M$ is a layered automaton with transitions of the form $(q_0,a,p)$ for $a\in A$ and $p\in X$
and $(p,b,q_f)$ for $b\in B$ and $p\in X$.
Hence, letting $A(p)$ to stand for the set $\{a\in A:(q_0,a,p)\in\delta\}$ and $B(p)$ stand for the set
$\{b\in B:(p,b,q_f)\in\delta\}$ we get that there is an associated biclique covering $\mathcal{C}_M$ of $L$ to $M$
consisting of the bicliques $K_{A(p),B(p)}$, $p\in X$ that has the same size as $M$.

Observe that the transformation is invertible in the sense that to each biclique covering $\mathcal{C}$ such an
automaton of the same size can be constructed, showing the second part of the claim.
\end{proof}

Since by Theorem~\ref{thm-orthogonal} there exist graphs with arbitrary large $\Rect(G)$ and $\Cov(G)=\Omega(\Rect(G)^{3/2-\epsilon})$,
we have the following as byproduct:
\begin{theorem}
\label{thm-epsilon}
For any $\epsilon>0$ and for arbitrarily large $n$ there exist languages (consisting of two-letter words only) which are recognizable
by NFAs of size $n$ but are only recognizable by $\varepsilon$-free NFAs of size $\Omega(n^{3/2-\epsilon})$.

In other words, in order to make an NFA $\varepsilon$-free, an $n^{3/2-\epsilon}$ blow-up in the size can be inevitable.
\end{theorem}

\section{Conclusion, future directions}
\label{sec-conclusion}

We proved a lower bound for the blow-up when transforming NFAs to $\varepsilon$-free NFAs.
We showed that the cost of $\varepsilon$-removal from NFAs is worst-case $\Omega(n^{3/2-\epsilon})$, improving the
previous bound $\Omega(n\log^2n)$.
The largest possible gap is between $\Omega(n^{3/2-\epsilon})$ and $O(n^2)$, just like in the case of going  from CFGs to chain rule free CFGs.
Narrowing these gaps seem to be nontrivial open problems.

We used a graph-theoretic approach by translating the problem into finding large blow-ups between two complexity measures for
bipartite graphs: the rectifier network size $\Rect$ and the minimal weight biclique covering $\Cov$.
We proved that there are graphs with arbitrarily large $\Rect$ value $n$ such that $\Cov=\Omega(n^{3/2-\epsilon})$ for any $\epsilon>0$.
We gave partial results for determining the largest possible blow-up between these quantities. These include a sufficient condition for
a subquadratic upper bound, and the sharpness of a combinatorial bound for the minimal weight biclique covering (obtained by proving a bound for the general
weighted set covering problem).
We also formulated two open problems about related combinatorial bounds, which appear to be of interest in themselves.
Solving these problems may also be useful for determining the largest possible blow-up.
The relationship between $\Rect$ and $\Cov$ can be viewed as a size-depth trade-off problem for depth-2 and unrestricted depth circuits computing sets of
Boolean disjunctions \cite{Wegener:1987:CBF:35517}.
As far as we know, there are many other related open problems, such as establishing a bounded-depth hierarchy.



{
\bibliographystyle{plain}
\bibliography{biblio}{}

\begin{thebibliography}{10}

\bibitem{BermanEtal}
Piotr Berman, Arnab Bhattacharyya, Elena Grigorescu, Sofya Raskhodnikova,
  David~P. Woodruff, and Grigory Yaroslavtsev.
\newblock Steiner transitive-closure spanners of low-dimensional posets.
\newblock {\em Combinatorica}, pages 1--24, 2014.

\bibitem{Blum1983287}
Norbert Blum.
\newblock More on the power of chain rules in context-free grammars.
\newblock {\em Theoretical Computer Science}, 27(3):287 -- 295, 1983.
\newblock Special Issue Ninth International Colloquium on Automata, Languages
  and Programming (ICALP) Aarhus, Summer 1982.

\bibitem{chvatal1979setcover}
V.~Chvatal.
\newblock A greedy heuristic for the set-covering problem.
\newblock {\em Mathematics of Operations Research}, 4(3):233--235, 1979.

\bibitem{Hromkovic:1997:TRE:646512.695338}
Juraj Hromkovic, Sebastian Seibert, and Thomas Wilke.
\newblock Translating regular expressions into small epsilon-free
  nondeterministic finite automata.
\newblock In {\em Proceedings of the 14th Annual Symposium on Theoretical
  Aspects of Computer Science}, STACS '97, pages 55--66, London, UK, 1997.
  Springer-Verlag.

\bibitem{Johnson:1973:AAC:800125.804034}
David~S. Johnson.
\newblock Approximation algorithms for combinatorial problems.
\newblock In {\em Proceedings of the Fifth Annual ACM Symposium on Theory of
  Computing}, STOC '73, pages 38--49, New York, NY, USA, 1973. ACM.

\bibitem{DBLP:books/daglib/0028687}
Stasys Jukna.
\newblock {\em Boolean Function Complexity - Advances and Frontiers}, volume~27
  of {\em Algorithms and combinatorics}.
\newblock Springer, 2012.

\bibitem{JuknaChapter2013}
Stasys Jukna.
\newblock Computational complexity of graphs.
\newblock In M.~Dehmer and F.~Emmert-Streib, editors, {\em Advances in Network
  Complexity}, pages 99--153. Wiley, 2013.

\bibitem{kollar1996norm}
J{\'a}nos Koll{\'a}r, Lajos R{\'o}nyai, and Tibor Szab{\'o}.
\newblock Norm-graphs and bipartite {T}ur{\'a}n numbers.
\newblock {\em Combinatorica}, 16(3):399--406, 1996.

\bibitem{LovaszThesis}
L\'aszl\'o Lov\'asz.
\newblock A kombinatorika minimax t\'eteleir{\H{o}}l.
\newblock {\em Matematikai Lapok}, 26:209--264, 1975.

\bibitem{Lovasz1975383}
L\'aszl\'o Lov\'asz.
\newblock On the ratio of optimal integral and fractional covers.
\newblock {\em Discrete Mathematics}, 13(4):383 -- 390, 1975.

\bibitem{Lupanov}
O.~B. Lupanov.
\newblock On rectifier and switching-and-rectifier schemes.
\newblock {\em Dokl. Akad. Nauk SSSR}, 111:1171--1174, 1956.

\bibitem{Mehlhorn}
Kurt Mehlhorn.
\newblock Some remarks on {B}oolean sums.
\newblock In Jir\'{\i} Becv{\'a}r, editor, {\em Mathematical Foundations of
  Computer Science 1979}, volume~74 of {\em Lecture Notes in Computer Science},
  pages 375--380. Springer Berlin Heidelberg, 1979.

\bibitem{conf/focs/MeyerF71}
Albert~R. Meyer and Michael~J. Fischer.
\newblock Economy of description by automata, grammars, and formal systems.
\newblock In {\em SWAT (FOCS)}, pages 188--191. IEEE Computer Society, 1971.

\bibitem{Nechiporuk}
E.~I. Nechiporuk.
\newblock On a {B}oolean matrix.
\newblock {\em Systems Theory Res.}, 21:236--239, 1971.

\bibitem{Pippenger:1976:SGA:321958.321962}
Nicholas Pippenger and Leslie~G. Valiant.
\newblock Shifting graphs and their applications.
\newblock {\em J. ACM}, 23(3):423--432, July 1976.

\bibitem{Schnitger06regularexpressions}
Georg Schnitger.
\newblock Regular expressions and {NFA}s without $\varepsilon$-transitions.
\newblock In {\em 23th Symposium on Theoretical Aspects of Computer Science
  (STACS 2006), LNCS 3884}, pages 432--443, 2006.

\bibitem{Stein1974391}
S.K Stein.
\newblock Two combinatorial covering theorems.
\newblock {\em Journal of Combinatorial Theory, Series A}, 16(3):391 -- 397,
  1974.

\bibitem{Wegener}
Ingo Wegener.
\newblock A new lower bound on the monotone network complexity of {B}oolean
  sums.
\newblock {\em Acta Informatica}, 13(2):109--114, 1980.

\bibitem{Wegener:1987:CBF:35517}
Ingo Wegener.
\newblock {\em The Complexity of {B}oolean Functions}.
\newblock John Wiley \& Sons, Inc., New York, NY, USA, 1987.

\bibitem{Williamson:2011:DAA:1971947}
David~P. Williamson and David~B. Shmoys.
\newblock {\em The Design of Approximation Algorithms}.
\newblock Cambridge University Press, New York, NY, USA, 1st edition, 2011.

\end{thebibliography}
}
\end{document}